\documentclass[10pt,journal]{IEEEtran}
\IEEEoverridecommandlockouts

\usepackage{cite}
\usepackage{amsmath}
\usepackage{amsthm}
\usepackage{amssymb}

\DeclareMathOperator*{\argmin}{arg\,min}
\usepackage{graphicx}
\usepackage{textcomp}
\usepackage{xcolor}
\usepackage{tabularray}
\usepackage{booktabs}
\usepackage{multirow}
\usepackage{makecell}

\usepackage{algorithm}
\usepackage[noend]{algpseudocode}
\makeatletter
\def\BState{\State\hskip-\ALG@thistlm}
\makeatother
\newcommand{\E}{\mathbb{E}}

\newtheorem{proposition}{Proposition}
\newtheorem{definition}{Definition}
\newtheorem{remark}{Remark}

\usepackage{matlab-prettifier}
\usepackage{color}
\usepackage{mathrsfs}
\usepackage[shortlabels]{enumitem}
\makeatletter
\def\BState{\State\hskip-\ALG@thistlm}
\makeatother
\def\BibTeX{{\rm B\kern-.05em{\sc i\kern-.025em b}\kern-.08em
    T\kern-.1667em\lower.7ex\hbox{E}\kern-.125emX}}
\usepackage{subcaption}

\begin{document}

\title{{\bf\Large Adaptive Incentive-Compatible Navigational Route Recommendations\\ in Urban Transportation Networks}
}

\author{Ya-Ting Yang, Haozhe~Lei, and Quanyan Zhu
\thanks{The Authors are with the Department of Electrical and Computer Engineering, New York University, Brooklyn, NY, 11201, USA; E-mail: {\tt\small \{yy4348, hl4155, qz494\}@nyu.edu}. Correspondence should be sent to YY.}%
\thanks{This work has been submitted to the IEEE for possible publication. Copyright may be transferred without notice, after which this version may no longer be accessible.}
}

\maketitle

\begin{abstract}
In urban transportation environments, drivers often encounter various path (route) options when navigating to their destinations. This emphasizes the importance of navigational recommendation systems (NRS), which simplify decision-making and reduce travel time for users while alleviating potential congestion for broader societal benefits. However, recommending the shortest path may cause the flash crowd effect, and system-optimal routes may not always align the preferences of human users, leading to non-compliance issues. It is also worth noting that universal NRS adoption is impractical. Therefore, in this study, we aim to address these challenges by proposing an incentive compatibility recommendation system from a game-theoretic perspective and accounts for non-user drivers with their own path choice behaviors. Additionally, recognizing the dynamic nature of traffic conditions and the unpredictability of accidents, this work introduces a dynamic NRS with parallel and random update schemes, enabling users to safely adapt to changing traffic conditions while ensuring optimal total travel time costs. The numerical studies indicate that the proposed parallel update scheme exhibits greater effectiveness in terms of user compliance, travel time reduction, and adaptability to the environment.
\end{abstract}


\section{Introduction}
Harnessing the vast information available from modern wireless communication and Internet of Things (IoT) advancements \cite{lv2020ai,zantalis2019review}, coupled with the progress made in data science and artificial intelligence \cite{veres2019deep,haydari2020deep}, intelligent transportation systems (ITS) have gained substantial attention for their ability to effectively tackle traffic congestion and elevate driver experiences through transportation management. However, despite the availability of real-time traffic information, navigating through complex urban environments or coping with congestion remains a challenging task for individuals. The complex urban environment naturally offers numerous path options for individuals to consider \cite{van2016user}, further complicating the decision-making process. This is where navigational recommendation systems (NRS) step in, including but not limited to platforms like Google Maps or Apple Maps. By offering path recommendations, NRS can help simplify decision-making for its users, at the same time reducing user travel times and mitigating potential congestion \cite{8262884}.

One of the most straightforward recommendations provided by NRS is to suggest the shortest path \cite{knuth1977generalization} in terms of travel time from the original starting point to the desired destination, based on current traffic conditions \cite{sherali1998time}. However, this approach can lead to a phenomenon known as the ``flash crowd effect'' \cite{flash_crowd_effect} if a large number of users are recommended the same shortest path at the same time. The flash crowd effect occurs when a sudden surge of users follows the same recommendation, causing a significant increase in traffic volume along that particular path. Consequently, the recommended path may become congested, resulting in longer travel times than initially anticipated for the shortest path. To address this issue, NRS must consider the collective impact of other users when providing recommendations. Instead of solely optimizing individual travel times, NRS can aim to optimize system efficiency by minimizing the total travel time costs across the urban transportation network. However, while prioritizing system efficiency seems beneficial for overall traffic management, it may come at the expense of some users' preferences and utilities. This trade-off can lead to user non-compliance issues, as users may choose not to follow the recommended path if it does not align with their preferences or if they discover alternative paths with lower travel time costs. As a result, achieving system optimality becomes challenging if users do not comply with the recommendations. 

\begin{figure}
    \centering
    \includegraphics[width=3.45in]{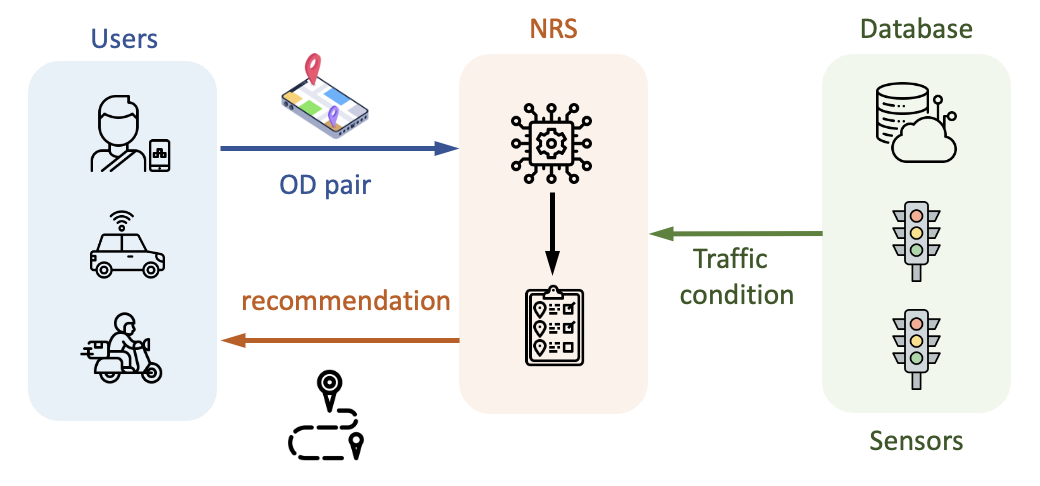}
    \caption{The process for navigational recommendations.}
    \label{fig:NRS_process}
\vspace{-3mm}
\end{figure}

Interacting directly with human users and being aware of potential non-compliance concerns, NRS must ensure that users lack incentives to unilaterally deviate from the recommendations, which coincides with the concept of Nash equilibrium. Consequently, this work suggests that a feasible recommendation provided by the NRS must satisfy the incentive constraints for the users. Furthermore, assuming universal adoption of a specific NRS among all individuals within the urban transportation network is impractical. Thus, this study also accounts for non-user drivers who navigate independently based on their own preferences and path choice behaviors.

Nevertheless, traffic conditions are dynamic and can change rapidly due to accidents, road closures, work zones for construction, weather conditions, and other factors. Therefore, merely planning before navigating on the road is not enough; the recommendations need to be updated according to the most recent traffic conditions. Thanks to recent advancements in edge computing and vehicle-to-everything (V2X) technologies, we can go one step further and consider a more dynamic navigational recommendation system (NRS with updated scheme). This study aims to utilize the edge computing framework, where the edge server (embedded in the road infrastructure) can collect localized data, exchange regional traffic conditions with other edge servers, and communicate with the local apps associated with NRS users. Then, the local NRS app can compute and provide the updated recommendation for its user. By incorporating real-time data and updates into navigational recommendations, users can receive more accurate and timely guidance to navigate efficiently and safely. Our contributions can be summarized as follows.
\begin{itemize}
    \item We introduce an incentive-compatible NRS from a game-theoretic perspective, considering both human user non-compliance issues and the behaviors of non-user drivers in the urban transportation network. Our objective is to guide users toward an optimal traffic equilibrium, ensuring that users have no incentive to deviate from the provided recommendations.
    \item We propose both parallel and random updated schemes for the NRS by leveraging V2X technology. It enables users to receive updates on navigational recommendations while driving, facilitating efficient navigation that helps users adapt to rapidly changing traffic conditions to reach their destinations promptly.
\end{itemize}

\noindent\textbf{Paper Organization:} The rest of this paper is organized as follows. Section 2 provides a brief overview of recent research applicable to path planning and navigational recommendations. Section 3 introduces the proposed navigational recommendation system, which takes user compliance and non-user behaviors into account. Section 4 explores a dynamic navigational recommendation system that adapts guidance based on evolving traffic conditions. Section 5 discusses the performance and paradox using numerical experiments under different scenarios. Finally, Section 6 concludes the paper.
 
\section{Literature Review}
In intelligent transportation systems, the emergence of smart mobile applications, including traffic information alerts, path planning tools, and navigational recommendations, has potential benefits for elevating user experiences, improving traffic management, and addressing safety issues \cite{siuhi2016opportunities}. Specifically, path planning and navigational recommendation applications are designed to assist travelers in efficiently navigating urban transportation networks, offering recommendations for determining the most effective path from their original starting point to their target destination. 

One aspect of recommendation systems often focuses on individually optimizing recommendations for each user by suggesting the shortest paths according to current traffic conditions \cite{knuth1977generalization,sherali1998time}. However, such individualized (selfish) recommendations tend to overlook the collective impact on other users, potentially leading to the flash crowd effect \cite{flash_crowd_effect}. That is, if a substantial number of users simultaneously navigate to the same shortest path, it can result in a surge of traffic volume along that path, increasing the overall congestion levels across the urban transportation network. This scenario embodies a classic example of the ``Tragedy of the Common'' in economics, wherein the pursuit of individual utilities inadvertently undermines the common good \cite{van2016user}, i.e., the shortest path in navigational recommendations. 

Alternatively, recommendation systems may prioritize overall system efficiency by minimizing the total travel time cost over the network. However, this approach can result in disparities where some users experience longer travel times compared to others traveling from the same origin to the same destination. This discrepancy raises concerns regarding user compliance and fairness. The challenge then lies in promoting socially responsible behavior among users, motivating them to accept recommendations that may entail higher travel costs in exchange for alleviating congestion for the collective benefit. 

On the contrary, recommendation systems can adopt the concept of user equilibrium, as seen in mixed-strategy or pure-strategy atomic routing games \cite{du2014distributed, du2015coordinated}. In such equilibrium scenarios, fairness among users is guaranteed, as individuals with the same origin and destination experience identical (minimum) travel times once equilibrium is reached. Following such a case, the user will not have incentives to deviate from the recommendation; thus, user compliance can be ensured. Several recent studies in equilibrium routing also share similar insights while incorporating machine learning techniques to improve traffic prediction or users' route choice behaviors. For example, \cite{mahajan2019design} utilizes neural networks for link travel time prediction while \cite{user_eqm_rl} uses the reinforcement learning method in which users learn to choose the optimal route based on their past experiences. However, the main drawback of a user equilibrium recommendation is the inability to minimize the total travel time across the entire transportation system.

Therefore, many research efforts aim to bridge the gap between system efficiency and user equilibrium \cite{VANESSEN2016527,morandi2023bridging}. They strive to minimize overall system costs associated with individually optimizing or user equilibrium recommendations to a certain desired level. Within this domain, two main research trends have emerged: one focuses on minimizing overall system costs while accounting for user preferences or selfish choice behavior constraints, and the other strategically designs or perturbs the information received by users to align user preferences with system efficiency \cite{8262884}. For instance, \cite{ning2023robust} aims to achieve system optimality while considering heterogeneous user preferences in choice behaviors and cooperation willingness; \cite{spana2021strategic} ensures that users' selfish route choices also contribute to improving system performance through strategically perturbed traffic information.

It is important to note that the recommendations to users in the intelligent transportation systems are not the same as either routing in the network \cite{Multihoming,zhang2010optimizing} systems or routing in transportation networks for connected autonomous vehicles \cite{rossi2018routing}. This differentiation stems from several key factors. Firstly, human users retain the freedom to not follow the recommendations. Secondly, not all drivers within urban transportation networks utilize navigational recommendation systems. Last but not least, the dynamic nature of traffic conditions introduces variability over time, necessitating adaptability beyond mere planning. Hence, our study aims to go one step further from routing. We explore human-centric considerations such as users' compliance with recommendations and the choice behavior of non-user drivers. Our proposed approach involves a comprehensive NRS planning framework designed to guide users toward a mixed Nash equilibrium. Furthermore, we tackle the challenge that comes from evolving traffic conditions by introducing a dynamic navigational recommendation system capable of updating recommendations based on changing traffic dynamics.

\section{The Navigational Recommendation System}\label{sec:RS_model}

In this section, we first establish a mathematical model for navigational recommendation systems (NRS) in urban transportation environments, including but not limited to popular platforms like Google Maps or Apple Maps in our daily lives. 
In addition, it is crucial to recognize that the NRS in urban transportation networks needs to be approached differently from routing in network systems, as NRS often involves user interaction, including but not limited to preferences, experiences, habits, etc. for example, drivers may favor more familiar paths, choose to deviate from the suggested paths due to congestion experienced during specific times in their past journeys or opt for alternative paths to avoid traffic lights along the recommended paths. 

\subsection{The System Model}
One of the essential components of NRS is the urban transportation network, represented by $\mathcal{G}=\{\mathcal{V}, \mathcal{E}\}$. In this context, the set of nodes $\mathcal{V}$ corresponds to intersections, while the set of edges $\mathcal{E}$ indicates the roads. Traveling along a road $e\in \mathcal{E}$ will incur a cost $c_e: \mathbb{R}_{\ge 0} \mapsto \mathbb{R}_{+}$ associated with the expected flow $f_e \in \mathbb{R}_{\ge 0}$ on that road $e$. One usual choice for the cost function $c_e(\cdot)$ is the travel time cost $$c_e(f_e)=t_e\left(1+\eta\left(\frac{f_e}{k_e}\right)^\zeta\right)$$ provided by the standard Bureau of Public Roads (BPR) function. Here, $t_e \in \mathbb{R}_{+}$ represents the free-flow travel time on road $e$, $k_e \in \mathbb{R}_{+}$ signifies the capacity of road $e$, and $\eta, \zeta \in \mathbb{R}_{\ge 0}$ are some parameters.

The user set of NRS is denoted as $\mathcal{U}$. Each user, $u \in \mathcal{U}$, is associated with a specific origin $O_u \in \mathcal{V}$ and destination $D_u \in \mathcal{V}$ pair. We refer to the pair as an OD pair, expressed by $\theta_u = (O_u, D_u)$, and the set of OD pairs for the NRS users is $\Theta_{\mathcal{U}} \subset \Theta$, with $\Theta =  |\mathcal{V}| \times |\mathcal{V}|$. Then, user $u$ with OD pair $\theta_u$ has the feasible path choice set $\mathcal{S}_u = \{s_{u, 1}, \cdots, s_{u, k_u}\}$. Each choice $s_{u, i} \in \mathcal{S}_u$ provides the user $u$ a path from the origin to the desired destination. 

This study explores the scenario where the NRS recommends a mixed strategy over feasible path choices to the users.
Define $\mathcal{P}_u:=\Delta \mathcal{S}_u$ as the simplex of $\mathcal{S}_u$ and  $\mathcal{P}:=\Pi_{u \in \mathcal{U}}\mathcal{P}_u$. A mixed strategy for user $u$ is $\textbf{P}_u \in \mathcal{P}_u$ so that $\textbf{P}_u=\{p_{u, i}\}_{i=1, \cdots, k_u}$ is a probability distribution over $\mathcal{S}_u$. Each $p_{u, i}\in[0,1]$ denotes the probability that the NRS recommends path $s_{u, i} \in \mathcal{S}_u$ to user $u$, subject to the constraints
$
    \sum_{i=1}^{k_u}p_{u, i}=1 , \ \forall u \in \mathcal{U}.
$ That is, $$\mathcal{P}_u := \left\{\textbf{P}_{u} \in \mathbb{R}^{k_u} \middle| p_{u,i}\geq 0, i=1, \cdots, k_u , \sum^{k_u}_{i=1}p_{u,i}=1 \right\},$$ which is compact and convex. Then, the recommendations suggested by the NRS to all users is $\textbf{P}=\{\textbf{P}_u\}_{u \in \mathcal{U}} \in \mathcal{P}$.

In transportation, from a microscopic perspective, the probability $p_{u, i}$ can be interpreted as the expected volume generated by user $u$ along path $s_{u, i}$. This, in turn, contributes to the expected flow (load) $f_e^r: \mathcal{P} \mapsto \mathbb{R}_{\ge 0}$ on edge $e \in \mathcal{E}$ as below. 
\begin{equation*}
    f_e^r(\textbf{P})=\sum_{u \in \mathcal{U}}\sum_{s_{u, i} \in \mathcal{S}_u}p_{u, i}\mathbf{1}_{\{e \in s_{u, i}\}},
\end{equation*} with the indicator function $\mathbf{1}_{\{e \in s_{u, i}\}}$ equal to $1$ if road $e$ is part of path $s_{u, i}$ and $0$ otherwise. Therefore, a generalized travel cost $C_{u, i}: \mathcal{P} \mapsto \mathbb{R}_{+}$ for user $u$ associated with path $s_{u, i}$ can be formulated by summing the costs of all the edges along the path:
\begin{equation*}
    C_{u, i}(\textbf{P})=\sum_{e \in s_{u, i}}c_e(f_e^r).
\end{equation*}

However, it is impractical to assume that all drivers on the roads are users of the NRS. In this context, those drivers who do not utilize the NRS are referred to as \textit{drivers}. While these drivers do not actively seek recommendations from the NRS, their behaviors still impact traffic conditions, influencing the flow on each road $e \in \mathcal{E}$. Thus, from a more comprehensive perspective, the NRS needs to incorporate the behaviors of these non-user drivers.

We denote the set of drivers who do not use the NRS as $\Bar{\mathcal{U}}$. Similar to the users, each driver $\Bar{u} \in \Bar{\mathcal{U}}$ is associated with an origin $O_{\Bar{u}} \in \mathcal{V}$ and destination $D_{\Bar{u}} \in \mathcal{V}$ pair, leading to a set of feasible paths $\mathcal{S}_{\Bar{u}} = \{s_{\Bar{u}, 1}, \cdots, s_{\Bar{u}, k_{\Bar{u}}}\}$. Without loss of generality, this work assumes that the stochastic choice behavior of these drivers can be modeled by the Multinomial Logit (MNL) model \cite{daskin1985urban}. The MNL model is widely used in various fields, including economics, transportation \cite{lee2018comparison}, marketing, and social sciences, to analyze discrete choice behaviors among multiple alternatives. 
In the context of path choices, the model supposes that individuals make choices based on the value they associate with each available option. That is, the behavior (mixed strategy) of driver $\Bar{u}$ over the set of feasible path choices $\mathcal{S}_{\Bar{u}}$ is given by
\begin{equation}
p_{\Bar{u}, i}^o=\frac{e^{V_{\Bar{u}, i}}}{\sum_{i=1}^{k_{\Bar{u}}}e^{V_{\Bar{u}, i}}}, \ \forall s_{\Bar{u}, i} \in \mathcal{S}_{\Bar{u}},
\label{eq:MNL}
\end{equation} where $V_{\Bar{u}, i}=-\alpha_{\Bar{u}}-\beta_{\Bar{u}} C_{{\Bar{u}}, i}^o$ represents the driver's value of path $s_{\Bar{u}, i}$, with $C_{\Bar{u}, i}^o \in \mathbb R_{+}$ being the current (initial) cost (travel time) on $s_{\Bar{u}, i}$. Here, $\alpha_{\Bar{u}} \in \mathbb{R}$ and $\beta_{\Bar{u}} \in \mathbb{R}$ are driver $\Bar{u}$'s preference parameters. 
We denote $\textbf{P}_{o}=\{p_{\Bar{u}, i}^o\}_{\Bar{u} \in \Bar{\mathcal{U}}, i=1, \cdots, k_{\Bar{u}}}$ as the mixed strategy played by the drivers. Then, the expected flow induced by other drivers on the road $e \in \mathcal{E}$ is $f_e^o: \Pi_{\Bar{u}\in \Bar{\mathcal{U}}} \Delta \mathcal{S}_{\Bar{u}} \mapsto \mathbb{R}_{\ge 0}$, where
\begin{equation*}
    f_e^o(\textbf{P}_o)=\sum_{\Bar{u} \in \Bar{\mathcal{U}}}\sum_{i=1}^{k_{\Bar{u}}}p^o_{\Bar{u}, i}\mathbf{1}_{\{e \in s_{\Bar{u}, i}\}}.
\end{equation*} In situations where the NRS takes into account other drivers, the combined expected flow caused by both users and drivers on each road $e \in \mathcal{E}$ is $f_e: \mathcal{P} \times \Pi_{\Bar{u}\in \Bar{\mathcal{U}}} \Delta \mathcal{S}_{\Bar{u}} \mapsto \mathbb{R}_{\ge 0}$, where $f_e(\textbf{P};\textbf{P}_{o}) = f_e^r(\textbf{P}) +f_e^o(\textbf{P}_{o})$, leading to the travel time cost $$c_e(f_e(\textbf{P};\textbf{P}_{o}))=t_e\left(1+\eta\left(\frac{f_e^r(\textbf{P})+f_e^o(\textbf{P}_{o})}{k_e}\right)^\zeta\right).$$ Thus, the travel cost for user $u$ associated with path $s_{u, i}$ becomes $\Tilde{C}_{u, i}: \mathcal{P} \times \Pi_{\Bar{u}\in \Bar{\mathcal{U}}} \Delta \mathcal{S}_{\Bar{u}} \mapsto \mathbb{R}_{+}$, and can be expressed as  
\begin{equation*}
    \Tilde{C}_{u, i}(\textbf{P};\textbf{P}_{o})=\sum_{e \in s_{u, i}}c_e(f_e^r(\textbf{P})+f_e^o(\textbf{P}_{o})),
\end{equation*} where the additional expected flow $f_e^o(\textbf{P}_{o})$ resulting from other drivers' path choice behaviors $\textbf{P}^o$ does not change with the recommended mixed strategies $\textbf{P}$.
To this end, the elements considered in the context of urban transportation network and taken into account by the NRS can be encapsulated using the notation $\mathscr{R}=\left\langle \  \mathcal{G}, (c_e(\cdot))_{e \in \mathcal{E}}, \mathcal{U}, (\mathcal{S}_u)_{u \in \mathcal{U}}, \mathcal{\Bar{U}}, (\mathcal{S}_{\Bar{u}})_{\Bar{u} \in \mathcal{\Bar{U}}} \ \right\rangle$, and we call $\mathscr{R}$ the ``NRS component''. 
The notations are summarized in Table \ref{tab:notation}.

\begin{table}[h!]
\centering
\caption{List of Notations}
 \begin{tabular}{|c l|} 
 \hline 
 Notation & Description \\  
 \hline\hline 
  $\mathcal{G}$ & urban transportation network \\ 
  $\mathcal{V}$ & set of intersections \\
  $\mathcal{E}$ & set of roads \\
  $\Theta$ & set of all the OD pairs within the network\\
  $c_e(\cdot)$ & cost function on road $e \in \mathcal{E}$ \\[0.5ex]
  \hline 
  $\mathcal{U}$ & set of users of NRS \\
  $\Theta_{\mathcal{U}}$ & set of users' OD pairs\\
  $\theta_u$ & OD pair $(O_u, D_u)$ for user $u \in \mathcal{U}$\\
  $\mathcal{S}_u$ & set feasible paths for user $u \in \mathcal{U}$ \\
  $\mathcal{P}_u$ & defined as the set of probability distributions over $\mathcal{S}_u$\\
  $\textbf{P}_u$ & recommended mixed strategy to user $u$, $\{p_{u, i}\}_{i=1, \cdots, k_u}$\\
  $\textbf{P}_{-u}$ & recommendations to other users except $u$\\
  $\mathcal{P}$ & defined as $\Pi_{u \in \mathcal{U}}\mathcal{P}_u$\\
  $\textbf{P}$ & recommendations suggested by the NRS to all users\\
  $\mathcal{\Bar{U}}$ & set of other drivers \\
  $\theta_{\Bar{u}}$ & OD pair $(O_{\Bar{u}}, D_{\Bar{u}})$ for driver $\Bar{u} \in \mathcal{\Bar{U}}$\\
  $\mathcal{S}_{\Bar{u}}$ & set feasible paths for driver $\Bar{u} \in \mathcal{\Bar{U}}$\\
  $\textbf{P}_{o}$ & mixed strategy played by drivers, $\{p_{\Bar{u}, i}^o\}_{\Bar{u} \in \Bar{\mathcal{U}}, i=1, \cdots, k_{\Bar{u}}}$\\
  $f_e^r(\cdot)$ & expected flow (load) from users on road $e \in \mathcal{E}$\\
  $f_e^o(\cdot)$ & expected flow (load) from drivers on road $e \in \mathcal{E}$\\
  $C_{u, i}(\cdot)$ & cost function for path $s_{u, i} \in \mathcal{S}_u$\\
  $\Tilde{C}_{u, i}(\cdot)$ & cost function for path $s_{u, i} \in \mathcal{S}_u$ considering other drivers\\
  $\mathscr{R}$ & NRS component\\
  $F_u(\cdot)$ & expected cost for user $u$, considering users other than $u$\\
  $\Tilde{F}_u(\cdot)$ & expected cost for user $u$, considering traffic conditions\\
  \hline
 \end{tabular}
\label{tab:notation}
\vspace{-5mm}
\end{table}

\subsection{Feasible Recommendations}
Navigational recommendations involve \textit{human users}, distinct from machines, computer programs, or protocols in network routing. Human users may choose not to adhere to the recommendation from the NRS if they find a better or more favorable alternative. Therefore, the recommended mixed strategy that the user will accept and follow needs to satisfy the following incentive-compatibility (IC) constraints, in which a user $u \in \mathcal{U}$
must prefer the recommended mixed strategy $\textbf{P}_u \in \mathcal{P}_u$. 
\begin{definition}[Incentive Compatibility (IC)]
    Considering an NRS component denoted as $\mathscr{R}$, a recommended mixed strategy $\textbf{P}_u \in \mathcal{P}_u$ for a user $u$ is called incentive compatible if the expected cost associated with the recommendation is the lowest, i.e., $ \forall \ \textbf{P}^{\prime}_{u} \in \mathcal{P}_u,$
    \begin{equation*}
        \sum^{k_u}_{i=1} p_{u,i}\Tilde{C}_{u,i}(\textbf{P}_{u},\textbf{P}_{-u}; \textbf{P}_o)-p^{\prime}_{u,i}\Tilde{C}_{u,i}(\textbf{P}^{\prime}_{u},\textbf{P}_{-u};\textbf{P}_o) \leq 0.
    \end{equation*}
\label{def:IC}
\vspace{-5mm}
\end{definition}
Note that in the definition, the recommendation to all users $\textbf{P} \in \mathcal{P}$ can be written as $\{\textbf{P}_{u}, \textbf{P}_{-u}\}$. Here, $\textbf{P}_{u} \in \mathcal{P}_u$ signifies the recommendation to user $u$, $\textbf{P}_{-u} \in \Pi_{u^{\prime} \in \mathcal{U}\setminus{u}}\mathcal{P}_{u^{\prime}}$ represents the recommendations to other users except $u$, and $\textbf{P}^{\prime}_{u}=\{p^{\prime}_{u,i}\}_{s_{u, i} \in \mathcal{S}_u} \in \mathcal{P}_u$ is other possible mixed strategy for user $u$ to deviate.
With Definition \ref{def:IC}, the NRS aims to find the \textit{feasible} recommendations to all the users, ensuring their adherence to the recommended strategies.
\begin{definition}[Feasible Recommendation]
    Considering an NRS component denoted as $\mathscr{R}$, a feasible recommendation profile to all the users $\textbf{P} \in \mathcal{P}$ needs to satisfy:
\begin{subequations}
  \begin{align}
    & \sum^{k_u}_{i=1}p_{u,i}\Tilde{C}_{u,i}(\textbf{P}_{u},\textbf{P}_{-u}; \textbf{P}_o)-p^{\prime}_{u,i}\Tilde{C}_{u,i}(\textbf{P}^{\prime}_{u},\textbf{P}_{-u}; \textbf{P}_o)\nonumber\\ 
    &\qquad \leq 0, 
      \ \forall \ \textbf{P}^{\prime}_{u} \in \mathcal{P}_u, \forall u \in \mathcal{U}, \label{eq:IC_cons}\\
    &\sum^{k_u}_{i=1}p_{u,i}=1, \forall u \in \mathcal{U}, \label{eq:prob_sum}\\
    & p_{u,i}\geq 0, \ \forall s_{u, i} \in \mathcal{S}_u, \forall u \in \mathcal{U}.\label{eq:prob_geq0}
  \end{align}
\label{prob:RS}
\end{subequations} 
\label{def:NRS}
\vspace{-5mm}
\end{definition}
The constraints given by \eqref{eq:prob_sum} and \eqref{eq:prob_geq0} guarantee that \( \textbf{P}_u \) constitutes a valid mixed strategy for all \( u \in \mathcal{U} \), also denoted as \( \textbf{P}_u \in \mathcal{P}_u, \ \forall u \in \mathcal{U} \). 
With Definition \ref{def:NRS}, it is worth noting that the NRS recommends mixed strategies that coincide with the concept of Nash equilibrium (NE), indicating that every user lacks an incentive to unilaterally deviate from the recommended mixed strategy. The NRS takes into account the choice behavior of the users as a group and creates coordinated incentive-compatible recommendations, differing from the approaches that recommend the system efficiency path or the shortest path to each user.

\subsection{A Game-Theoretic Interpretation}

The NRS's problem of finding feasible recommended mixed strategies for all users can also be interpreted using a non-cooperative game. With the NRS component $\mathscr{R}$ and the expected cost $F_u: \mathcal{P} \times \Pi_{\Bar{u}\in \Bar{\mathcal{U}}} \Delta \mathcal{S}_{\Bar{u}} \mapsto \mathbb{R}_{\ge 0}$ of user $u$, where
\begin{equation}
    F_u(\textbf{P}_{u}, \textbf{P}_{-u}; \textbf{P}_o):=\sum^{k_u}_{i=1}p_{u,i}\Tilde{C}_{u,i}(\textbf{P}_{u},\textbf{P}_{-u}; \textbf{P}_o),
\label{eq:F_u}
\end{equation} a non-cooperative game $\Gamma$ is defined as $\Gamma:=(\mathscr{R}, (F_u)_{u \in \mathcal{U}})$. Each user $u \in \mathcal{U}$ of the NRS is a player of the game $\Gamma$ who aims to minimize his/her own expected cost by deciding a mixed strategy $\textbf{P}_u \in \mathcal{P}_u$ over feasible path choice set $\mathcal{S}_u$ given other users' strategies $\textbf{P}_{-u}$. That is, user $u$ in the game $\Gamma$ plays the best response to other users' strategies $\textbf{P}_{-u}$.
\begin{equation}
\begin{aligned}
    \text{OP}_u: \min_{\textbf{P}_{u}^{\prime}} \ &F_u(\textbf{P}_{u}^{\prime}, \textbf{P}_{-u}; \textbf{P}_o)\\
    \text{s.t. } 
    & \textbf{P}_u^{\prime} \in \mathcal{P}_u.
\end{aligned}
\label{eq:OP_u} 
\end{equation}

Then, the solution concept that coincides with Problem \eqref{prob:RS} in Definition \ref{def:NRS} is the mixed-strategy Nash equilibrium, where no player can reduce his/her own expected cost by deviating from the recommended mixed strategy.
\begin{definition}[Nash Equilibrium (NE) of Game $\Gamma$]
    Consider a non-cooperative game $\Gamma=(\mathscr{R}, (F_u)_{u \in \mathcal{U}})$, the mixed-strategy profile $\textbf{P}^* \in \mathcal{P}$ is a Nash equilibrium of the game if it satisfies
    \begin{equation}
        F_u(\textbf{P}_{u}^*, \textbf{P}_{-u}^*; \textbf{P}_o) \leq F_u(\textbf{P}_{u}^{\prime}, \textbf{P}_{-u}^*; \textbf{P}_o), \ \forall \textbf{P}_{u}^{\prime} \in \mathcal{P}_u, \ \forall u \in \mathcal{U}.
    \label{eq:VE1}
    \end{equation}
\label{def:NE}
\vspace{-5mm}
\end{definition} That is, $\textbf{P}^*$ constitutes an NE if, for all $u \in \mathcal{U}$, the strategy $\textbf{P}_u^*$ optimally solves the individual optimization problem $\text{OP}_u$ when all other users adopt their equilibrium strategies $\textbf{P}_{-u}^*$.
Since problems $\text{OP}_u, \ \forall u \in \mathcal{U}$ are constrained optimization problems, we can use projected gradient descent (PGD) \cite{boyd2004convex} or other optimization techniques to solve them. The objective of $\text{OP}_u$ under the equilibrium strategy profile $\textbf{P}^*$, denoted as $F_u(\textbf{P}_{u}^*, \textbf{P}_{-u}^*; \textbf{P}_o)$, is referred to as the ``NRS value'' of user $u$.

\begin{proposition}
    The feasible set of recommendations specified in Definition \ref{def:NRS} is equivalent to the set of NE of the game $\Gamma$, defined in Definition \ref{def:NE}. 
\label{prop:2to4}
\end{proposition}
\begin{proof}
    The equivalence is proved by showing each of the two sets is a subset of the other set. Denote the set of recommendations described by \eqref{prob:RS} as $\mathcal{R}$ and the set specified by \eqref{eq:VE1} as $\mathcal{N}$. Noting that constraints \eqref{eq:prob_sum} and \eqref{eq:prob_geq0} are equivalent to $ \textbf{P} \in \mathcal{P}$, we begin by proving that given $\mathcal{N}$, every element of $\mathcal{R}$ is also in $\mathcal{N}$ ($\mathcal{R}$ is a subset of $\mathcal{N}$). This can be shown by contradiction. Suppose there exists an element $\textbf{P}$ in $\mathcal{R}$ that does not satisfy \eqref{eq:VE1}, with the expected cost specified in \eqref{eq:F_u}. This implies that at least one user $u$ can unilaterally deviate from the recommended $\textbf{P}_u$ to reduce the expected cost, contradicting the assumption that $\textbf{P}$ is an element of $\mathcal{R}$ where each $\textbf{P}_u$ minimizes the expected cost of user $u$ given $\textbf{P}_{-u}$. Then, we prove that given the set of $\mathcal{R}$, every element of $\mathcal{N}$ is also in $\mathcal{R}$. This is also be shown by contradiction. If there exists an element $\textbf{P}'$ of $\mathcal{N}$ that does not satisfy \eqref{prob:RS}, it means that there exists at least one user $u$ whose expected cost can be further reduced, contradicting the assumption that $\textbf{P}$ is an NE that users can not be better-off through unilateral deviation. These steps complete the proof.
\end{proof}

\subsection{Existence of Feasible Recommendation}

We use the road-path incidence matrix $A_{|\mathcal{E}|\times |\Pi_{u \in \mathcal{U}} \mathcal{S}_u|}=[a_{es_{u, i}}]$ to depict the relationship between roads and paths. Each element is defined as follows.
$$ a_{es_{u, i}}=
\begin{cases}
    1 \qquad \text{if} \ e \in s_{u, i},\\
    0 \qquad \text{otherwise}.
\end{cases}
$$ The expected flow (load) caused by all users on the road $e \in \mathcal{E}$ can also be expressed as
\begin{equation*}
    f_e^r(\textbf{P})=\sum_{u \in \mathcal{U}}\sum_{s_{u, i} \in \mathcal{S}_u}p_{u, i}a_{es_{u, i}},
\end{equation*} which consists of two parts: $f_e^u(\textbf{P}_u) = \sum_{s_{u, i} \in \mathcal{S}_u}p_{u, i}a_{es_{u, i}}$ resulting from user $u$, and $f_e^{-u}(\textbf{P}_{-u})=f_e^r-f_e^u$ representing the part contributed by users other than $u$. Let the aggregated flow resulting from users other than $u$ and drivers $$\hat{f}_e^{-u}(\textbf{P}_{-u};\textbf{P}_o)=f_e^{-u}(\textbf{P}_{-u})+f_e^o(\textbf{P}_o), \forall e \in \mathcal{E}.$$ If we consider the BPR function for the cost $c_e(\cdot), \forall e \in \mathcal{E}$, then the expected cost for user $u$ is expressed as:
\begin{equation*}
\begin{aligned}
    & F_u(\textbf{P}_{u}, \textbf{P}_{-u}; \textbf{P}_o) \\
& =\sum_{i=1}^{k_u} p_{u, i}\sum_{e \in s_{u, i}}c_e(f_e^u(\textbf{P}_u)+f_e^{-u}(\textbf{P}_{-u})+f_e^o(\textbf{P}_{o}))\\
& =\sum_{i=1}^{k_u} p_{u, i} \sum_{e \in s_{u, i}} t_e\left[1+\eta\left(\frac{\hat{f}_e^{-u}(\textbf{P}_{-u};\textbf{P}_o)+f_e^u(\textbf{P}_u)}{k_e}\right)^\zeta \right],
\end{aligned}
\end{equation*} For simplicity, the value of $\hat{f}_e^{-u}(\textbf{P}_{-u};\textbf{P}_o)$ is represented by $\hat{f}_e^{-u}$ for subsequent analyses. Then, we can compute the gradient of $F_u$ with respect to each element $p_{u, i}$ of $\textbf{P}_{u}$. By noting the fact that $\partial f_e^u/\partial p_{u, i}=a_{es_{u, i}}$ and $\partial \hat{f}_e^{-u}/\partial p_{u, i}=0$, we have each 
\begin{equation*}
\begin{aligned}
    \frac{\partial F_u}{\partial p_{u, i}} &= \Tilde{C}_{u, i}(\textbf{P}_u, \textbf{P}_{-u};\textbf{P}_o)+\sum_{j=1}^{k_u}p_{u, j}\frac{\partial \Tilde{C}_{u, j}(\textbf{P}_u, \textbf{P}_{-u};\textbf{P}_o)}{\partial p_{u, i}}\\
    & =\sum_{e \in \mathcal{E}} a_{es_{u, i}} t_e\bigg[1+\eta\bigg(\frac{\hat{f}_e^{-u}+f_e^u(\textbf{P}_u)}{k_e}\bigg)^\zeta\\
    & \qquad \qquad \qquad + \zeta f_e^u(\textbf{P}_u) \frac{\eta}{k_e}\bigg(\frac{\hat{f}_e^{-u}+f_e^u(\textbf{P}_u)}{k_e}\bigg)^{\zeta-1}\bigg].\\
\end{aligned}
\end{equation*} The partial derivative $\partial F_u/\partial p_{u, i}$ depends on $\textbf{P}_u$ as well as the aggregated flow $\hat{f}_e^{-u}$ resulting from all other users and drivers.
\begin{equation*}
\begin{aligned}
    \frac{\partial F_u}{\partial p_{u', i}} &=\sum_{j=1}^{k_{u}}p_{u, j}\frac{\partial \Tilde{C}_{u, j}(\textbf{P}_u, \textbf{P}_{-u};\textbf{P}_o)}{\partial p_{u', i}}\\
    & = \sum_{e \in \mathcal{E}} a_{es_{u', i}} t_e \bigg[ \zeta f_e^u(\textbf{P}_u) \frac{\eta}{k_e}  \bigg(\frac{\hat{f}_e^{-u}+f_e^u(\textbf{P}_u)}{k_e}\bigg)^{\zeta-1}\bigg].
\end{aligned}
\end{equation*}
Therefore, the expected cost $F_u$ for user $u$ is continuously differentiable in $\textbf{P}=\{\textbf{P}_{u}, \textbf{P}_{-u}\} \in \mathcal{P}$. 

\begin{figure*}
    \centering
    \includegraphics[width=5.4in]{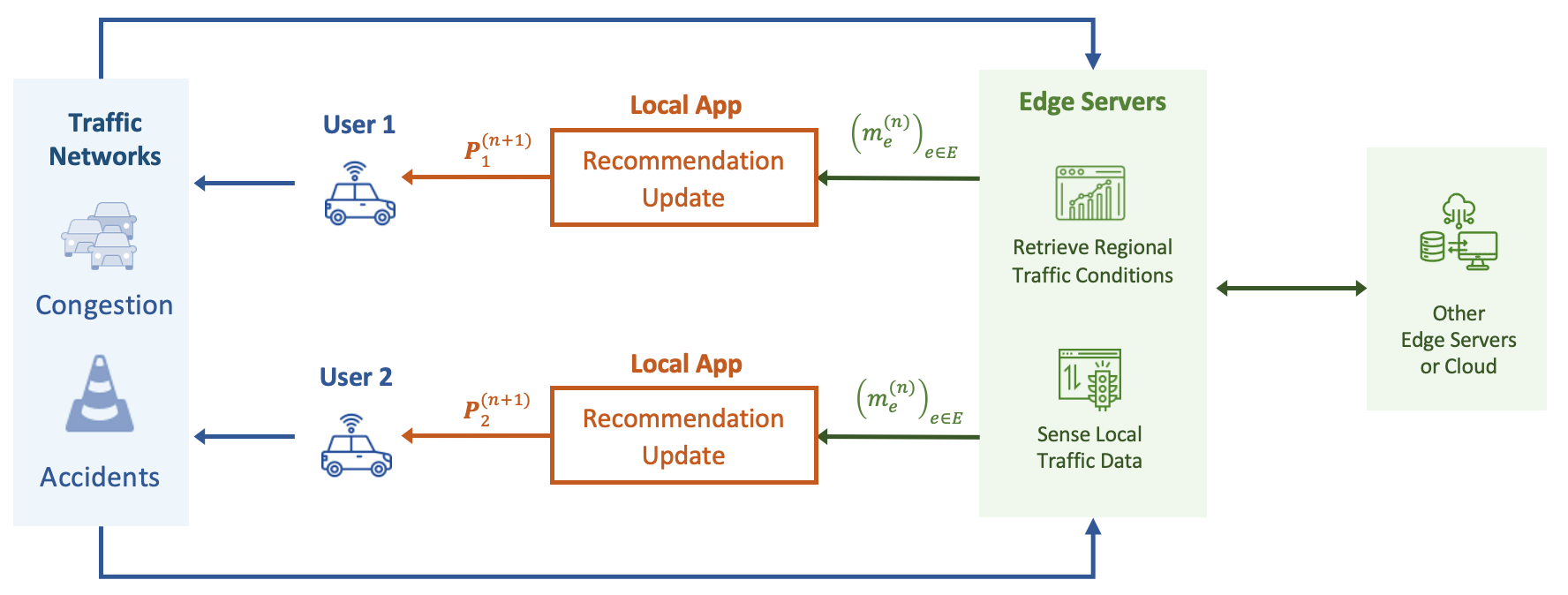}
    \caption{The (parallel) updated scheme by utilizing the V2X technology, such as edge server in the road infrastructure. The edge server is capable of gathering traffic data from its nearby roads and exchanging regional traffic conditions with either the cloud or other edge servers. Then, the local NRS app associated with each user computes the updated recommendation based on current traffic conditions provided by the edge server.}
    \label{fig:update_workflow}
\vspace{-1mm}
\end{figure*}

To check the convexity of $F_u$ in $\textbf{P}_u$, the Hessian matrix $H^{F_u}$ of $F_u$ needs to be positive semi-definite for all $\textbf{P}_u \in \mathcal{P}_u$. Note that each diagonal entry $h^{F_u}_{ii}$ in the $i$-th row and $i$-th column is
\begin{equation*}
\begin{aligned}
    h^{F_u}_{ii}&=\frac{\partial^2 F_u}{\partial p_{u, i}\partial p_{u, i}}\\
    &= \eta \sum_{e\in \mathcal{E}}  a_{es_{u, i}} \bigg[2 \zeta  \frac{t_e}{k_e} \bigg(\frac{\hat{f}_e^{-u}+f_e^u(\textbf{P}_u)}{k_e}\bigg)^{\zeta-1}\\
    & \quad + \zeta(\zeta - 1)f_e^u(\textbf{P}_u) \frac{t_e}{k_e^2}\bigg(\frac{\hat{f}_e^{-u}+f_e^u(\textbf{P}_u)}{k_e}\bigg)^{\zeta-2}\bigg] \geq 0,
\end{aligned}
\end{equation*} and each non-diagonal entry $h^{F_u}_{ij}$ in the $i$-th row and $j$-th column is
\begin{equation*}
\begin{aligned}
    h^{F_u}_{ij}& =\frac{\partial^2 F_u}{\partial p_{u, i}\partial p_{u, j}} \\
    &= \eta \sum_{e\in \mathcal{E}}  a_{es_{u, i}} a_{es_{u, j}} \bigg[2 \zeta  \frac{t_e}{k_e} \bigg(\frac{\hat{f}_e^{-u}+f_e^u(\textbf{P}_u)}{k_e}\bigg)^{\zeta-1}\\
    & \quad + \zeta(\zeta - 1)f_e^u(\textbf{P}_u) \frac{t_e}{k_e^2}\bigg(\frac{\hat{f}_e^{-u}+f_e^u(\textbf{P}_u)}{k_e}\bigg)^{\zeta-2}\bigg] \geq 0.
\end{aligned}
\end{equation*} For the positive semi-definiteness, we need the notion of diagonal dominance, stating that 
$$
|h^{F_u}_{ii}| \geq \sum_{j \neq i} |h^{F_u}_{ij}|, \forall i.
$$ Then, a diagonally dominant square matrix with real non-negative diagonal entries is positive semi-definite.

When $\zeta=1$, we can observe that the diagonal entries $h^{F_u}_{ii}$ sum the road-specific non-negative value $t_e/k_e$ over the road $e$ contained in the path $s_{u, i}$, $e \in s_{u_i}$, while the non-diagonal entries $h^{F_u}_{ij}$ sum the value $t_e/k_e$ over the road $e$ that intersects path $s_{u, i}$ and $s_{u, j}$, $e \in (s_{u,i} \cap s_{u, j})$. That is, if $\sum_{i=1}^{k_u} a_{es_{u, i}}\leq 2, \ \forall e \in \mathcal{E}$, the matrix $H^{F_u}$ is diagonally dominant. It may require more detailed analyses for cases where $\zeta>1$.

According to \cite{scutari2010convex}, if $F_u$ is convex in $\textbf{P}_u, \forall u \in \mathcal{U}$, by denoting gradient (vector) of $F_u$ as $\nabla_u F_u = (\partial F_u/\partial p_{u, i})_{s_{u, i} \in \mathcal{S}_u}$ and gradient of $F_u$ at $\textbf{P}_u, \textbf{P}_{-u}, \textbf{P}_o$ as $\nabla_u F_u(\textbf{P}_u, \textbf{P}_{-u};\textbf{P}_o)$. Let $\nabla F(\textbf{P};\textbf{P}_o)=(\nabla_u F_u(\textbf{P}_u, \textbf{P}_{-u};\textbf{P}_o))_{u \in \mathcal{U}}$, and the solution concept (NE) in Definition \ref{def:NE} is equivalent to the variational inequality (VI) characterization $\text{VI}(\mathcal{P}, \nabla F(\textbf{P};\textbf{P}_o))$ in finding $\textbf{P}^* \in \mathcal{P}$ as follows:
\begin{equation}
     (\textbf{P}-\textbf{P}^*)^\top \nabla F(\textbf{P}^*;\textbf{P}_o) \geq 0, \ \forall \textbf{P} \in \mathcal{P}.
\end{equation} Based on the existence results from VI (see (18) in \cite{scutari2010convex}), the game $\Gamma$ has a nonempty solution set, indicating that the set of feasible recommendations in Definition \ref{def:NRS} is nonempty.

\subsection{Projected Gradient Descent (PGD) Method} \label{sec:PGD}
Following Proposition \ref{prop:2to4}, we consider the projected gradient descent method in this section. The gradient step for user $u$ who aims to solve $\textbf{OP}_u$ given $\textbf{P}_{-u}$ from other users can be expressed as 
\begin{equation*}
    \textbf{q}_u^{(n+1)} = \textbf{P}_u^{(n)} - \alpha_n \nabla_u F_u(\textbf{P}_u^{(n)}, \textbf{P}_{-u};\textbf{P}_o),
\end{equation*} and then project it onto the simplex $\mathcal{P}_u$ as
\begin{equation*}
    \textbf{P}_u^{(n+1)} = \argmin_{\textbf{P}_u \in \mathcal{P}_u} ||\textbf{P}_u - \textbf{q}_u^{(n+1)}|| = \text{proj}_{\mathcal{P}_u} \left[\textbf{q}_u^{(n+1)}\right].
\end{equation*} 
The minimizers of $\textbf{OP}_u$ (user $u$'s best response to $\textbf{P}_{-u}$) are fixed points $\textbf{P}_u^{\dagger} \in \mathcal{P}_u$ of the projected gradient update
\begin{equation*}
    \textbf{P}_u^{\dagger} = \text{proj}_{\mathcal{P}_u}\left[\textbf{P}_u^{\dagger} - \alpha \nabla_u F_u(\textbf{P}_u^{\dagger}, \textbf{P}_{-u};\textbf{P}_o)\right].
\end{equation*} To this end, according to Definition \ref{def:NE} and Proposition \ref{prop:2to4}, we know that an NE is equivalent to the concatenation of strategies $\textbf{P}_u^*$, $\forall u \in \mathcal{U}$, that optimally solve the individual optimization problem $\text{OP}_u$, $\forall u \in \mathcal{U}$, when all other users adopt their equilibrium strategies $\textbf{P}_{-u}^*$.
\begin{equation}
    \textbf{P}_u^{*} = \text{proj}_{\mathcal{P}_u} \left[\textbf{P}_u^{*} - \alpha \nabla_u F_u(\textbf{P}_u^{*}, \textbf{P}_{-u}^{*};\textbf{P}_o )\right], \ \forall u \in \mathcal{U}.
\label{eq:PGD}
\end{equation}
\begin{remark}
    According to the discussion above, the set of NE of the game $\Gamma$, as defined in Definition \ref{def:NE}, is equivalent to the set of fixed points of \eqref{eq:PGD}. Hence, we can find the NE by finding a fixed point to \eqref{eq:PGD}.
\end{remark}

\section{The Updated Scheme}

With recent advances in vehicle-to-everything (V2X) technology and edge computing, each user (or the local app associated with each user) is capable of getting the current traffic conditions through edge servers embedded in the road infrastructure. Denote traffic conditions at time step $n$ as $\textbf{m}^{(n)} \in \mathbb{R}_{\ge 0}^{|\mathcal{E}|}$, with $\textbf{m}^{(n)}= (m_e^{(n)})_{e \in \mathcal{E}}$. Recall that each element, represented by $m_e^{(n)}$, comes from the resulting expected flow and define $M_e: \mathcal{P} \times \Pi_{\Bar{u}\in \Bar{\mathcal{U}}} \Delta \mathcal{S}_{\Bar{u}} \mapsto \mathbb{R}_{\ge 0}$ so that 
$$m_e^{(n)}=M_e(\textbf{P}^{(n)};\textbf{P}_o^{(n)})=\frac{f_e^u(\textbf{P}_u^{(n)})+\hat{f}_e(\textbf{P}_{-u}^{(n)};\textbf{P}_o^{(n)})}{k_e}.$$
After gathering the current traffic conditions $\textbf{m}^{(n)}$, user $u$ can utilize his/her navigational recommendation $\textbf{P}_u^{(n)}$ at time step $n$ to infer the conditions influenced by other users and drivers, denoted as $(\hat{f}_e(\textbf{P}_{-u}^{(n)};\textbf{P}_o^{(n)}))_{e \in \mathcal{E}}$. It is worth noting that the PGD method described in Section \ref{sec:PGD} for solving $\textbf{OP}_u$ requires user $u$ to have access to both the recommendations for other users $\textbf{P}_{-u}$, and the drivers' behaviors $\textbf{P}_{o}$. However, once the values of $\textbf{m}^{(n)}$ become available to users, local apps (or users) can directly utilize $\textbf{m}^{(n)}$ to update the recommendations. We will discuss two update schemes in the following subsections.

\subsection{Parallel Update Algorithm}
In this section, we explore a parallel update scheme where every user (or local app) synchronously adjusts their recommendation at each time step. The updated scheme is illustrated in Fig. \ref{fig:update_workflow}. By leveraging the traffic condition $\textbf{m}^{(n)}$, we can rewrite the cost function evaluated by user $u$ as $\Tilde{F}_u: \mathcal{P}_u \times \mathbb{R}_{\ge 0}^{|\mathcal{E}|} \mapsto \mathbb{R}_{+}$, where at time step $n$, 
$$\Tilde{F}_u(\textbf{P}_u^{(n)},\textbf{m}^{(n)})=\sum_{i=1}^{k_u} p_{u, i}^{(n)} \sum_{e \in \mathcal{E}}a_{es_{u,i}}t_e\left[1+\eta m_e^{\zeta (n)}\right].$$
Note that $\Tilde{F}_u(\textbf{P}_u^{(n)},\textbf{m}^{(n)})$ and $F_u(\textbf{P}_{u}^{(n)}, \textbf{P}_{-u}^{(n)}; \textbf{P}_o^{(n)})$ have identical value at the same $\textbf{P}_{u}^{(n)}, \textbf{P}_{-u}^{(n)}, \textbf{P}_o^{(n)}$.
Then, the gradient of $\Tilde{F}_u$ with respect to each element $p_{u, i}^{(n)}$ of $\textbf{P}_u^{(n)}$ becomes $\partial \Tilde{F}_u/\partial p_{u, i}^{(n)}=$
\begin{equation*}
\begin{aligned}
    &\sum_{e \in \mathcal{E}} a_{es_{u, i}} t_e\left[1+\eta\Big(m_e^{\zeta (n)} + \zeta  \frac{f_e^u(\textbf{P}_u^{(n)})}{k_e} m_e^{\zeta-1 (n)}\Big)\right].
\end{aligned}
\end{equation*}
Let $\nabla_u \Tilde{F}_u =(\partial \Tilde{F}_u/\partial p_{u, i}^{(n)})_{s_{u, i} \in \mathcal{S}_u}$ and denote gradient of $\Tilde{F}_u$ at $\textbf{P}_u^{(n)}, \textbf{m}^{(n)}$ as $\nabla_u \Tilde{F}_u(\textbf{P}_u^{(n)}, \textbf{m}^{(n)})$, starting with some initial point for all users, $\textbf{P}^{(0)} \in \mathcal{P}$, the updated recommendation at time step $n+1$ is, therefore, $\forall u \in \mathcal{U}$,
\begin{equation}
    \textbf{P}_u^{(n+1)} = \text{proj}_{\mathcal{P}_u}\left[\textbf{P}_u^{(n)}-\alpha_n \nabla_u \Tilde{F}_u(\textbf{P}_u^{(n)}, \textbf{m}^{(n)})\right].
\label{eq:update_P_u}
\end{equation} By denoting $\nabla \Tilde{F}(\textbf{P}^{(n)},\textbf{m}^{(n)})=(\nabla_u \Tilde{F}_u(\textbf{P}_u^{(n)}, \textbf{m}^{(n)}))_{u \in \mathcal{U}}$ and concatenating \eqref{eq:update_P_u} together, we can get:
\begin{equation}
    \textbf{P}^{(n+1)} = \text{proj}_{\mathcal{P}}\left[\textbf{P}^{(n)}-\alpha_n \nabla \Tilde{F}(\textbf{P}^{(n)}, \textbf{m}^{(n)})\right].
\label{eq:update_P}
\end{equation}  It is straightforward to see that the $u$-th component of $\textbf{P}^{(n+1)}$ given by \eqref{eq:update_P} is equivalent to the one in \eqref{eq:update_P_u}. 

As shown in Fig. \ref{fig:update_workflow}, the edge server senses local traffic data from nearby roads and exchange regional traffic conditions with other edge servers (or get regional data from the cloud); the local app obtains current traffic conditions, perform gradient updates according to \eqref{eq:update_P_u}, and then provide the recommendation to the user.
We will use the form stated in \eqref{eq:update_P} for convergence analysis.
\begin{proposition}
    Under the conditions that for all $u \in \mathcal{U}$, the expected cost $F_u$ is continuously differentiable in $\textbf{P} \in \mathcal{P}$ and convex in $\textbf{P}_u \in \mathcal{P}_u$, the sequence of recommendations $\{\textbf{P}^{(n)}\}$, or equivalently the strategy profile of users in the game $\Gamma$, generated by \eqref{eq:update_P} converges to an NE, $\textbf{P}^{*}$, defined in Definition \ref{def:NE}.
\label{prop:converge}
\end{proposition}
\begin{proof}
    The proof is provided in Appendix \ref{app:prop_converge}.
\end{proof}
Therefore, with Proposition \ref{prop:converge}, the $u$-th component of $\textbf{P}^{(n+1)}$ can be performed locally by the NRS app linked to user $u$. Such distributed gradient updates still lead to the convergence to an NE \cite{facchinei2003finite,nguyen2022distributed,mirror2023} under mild assumptions.

\subsection{Random Update Algorithm}
In practice, not all users receive the traffic conditions nor update the navigational recommendations simultaneously. Therefore, we consider the random update algorithm in this section. Each user updates his/her recommendations in discrete time intervals and infinitely often with a predefined probability $0 < \pi_u < 1$. That is, 
\begin{equation}
    \textbf{P}_u^{(n+1)} = 
    \begin{cases}
\text{proj}_{\mathcal{P}_u}\left[\textbf{P}_u^{(n)}-\alpha_n \nabla_u \Tilde{F}_u(\textbf{P}_u^{(n)}, \textbf{m}^{(n)})\right],  \text{w.p.} \ \pi_u\\
        \textbf{P}_u^{(n)}, \ \text{w.p.}\ 1-\pi_u.
    \end{cases}
\label{eq:update_P_u_r}
\end{equation}
The probability of update $\pi_u$ may vary based on the user's driving habits or the capabilities of V2X technologies within the region containing the user's origin and destination. Consequently, at each time step $n$, only a randomly selected set of users receive the updated recommendations. The random update scheme converges almost surely to an NE under mild assumptions.
\begin{proposition}
    Under the conditions that for all $u \in \mathcal{U}$, the expected cost $F_u$ is continuously differentiable in $\textbf{P} \in \mathcal{P}$ and convex in $\textbf{P}_u \in \mathcal{P}_u$. The sequence of recommendations $\{\textbf{P}^{(n)}\}$, or equivalently the strategy profile of users in the game $\Gamma$, generated by the random update \eqref{eq:update_P_u_r}, converges a.s. to an NE, $\textbf{P}^*$, defined in Definition \ref{def:NE}.
\label{prop:converge2}
\end{proposition}
\begin{proof}
    We give the proof in Appendix \ref{app:prop_converge2}.
\end{proof}



\section{Numerical Studies}

In this section, we illustrate our proposed NRS via numerical studies based on a network depicted in Fig. \ref{fig:network1} and a more complex one shown in Fig. \ref{fig:network2}.

\subsection{Baselines}
To illustrate the enhanced efficacy of the proposed NRS in reducing user/total travel time and its adaptability to changing environments, we compare the outcomes of the following baseline with the proposed updated recommendation.
\subsubsection{Selfish recommendation} We refer to the case of simply suggesting the shortest path without considering interactions with other users and drivers as selfish recommendations. More specifically, when a user submits a recommendation request with the OD pair, the recommendation system directly recommends the shortest travel time path to the user based solely on current traffic conditions.

\begin{figure}
    \centering
    \includegraphics[width=3.45in]{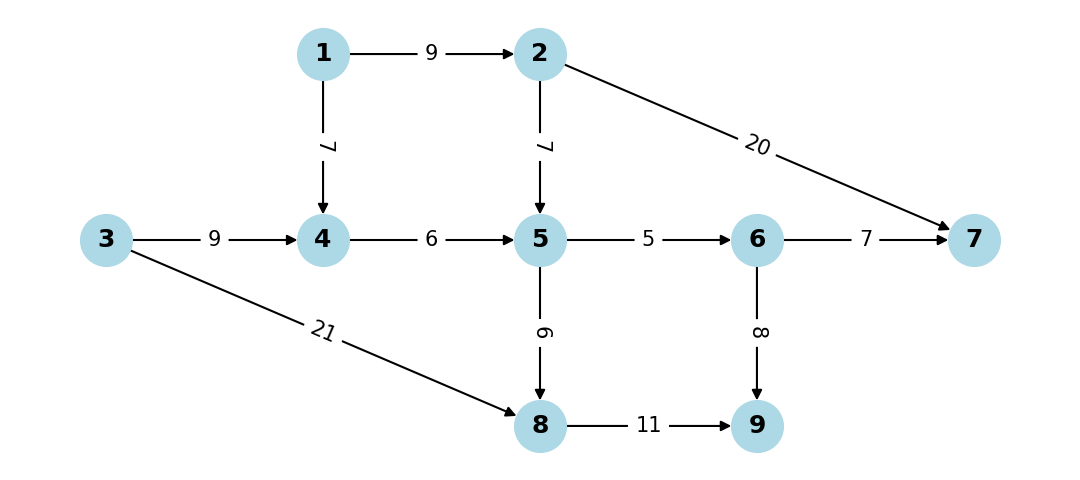}
    \caption{An example network $1$ for our numerical study, where the value on each edge denotes the free-flow travel time.}
    \label{fig:network1}
\vspace{-3mm}
\end{figure}

\begin{figure}
    \centering
    \includegraphics[width=3.45in]{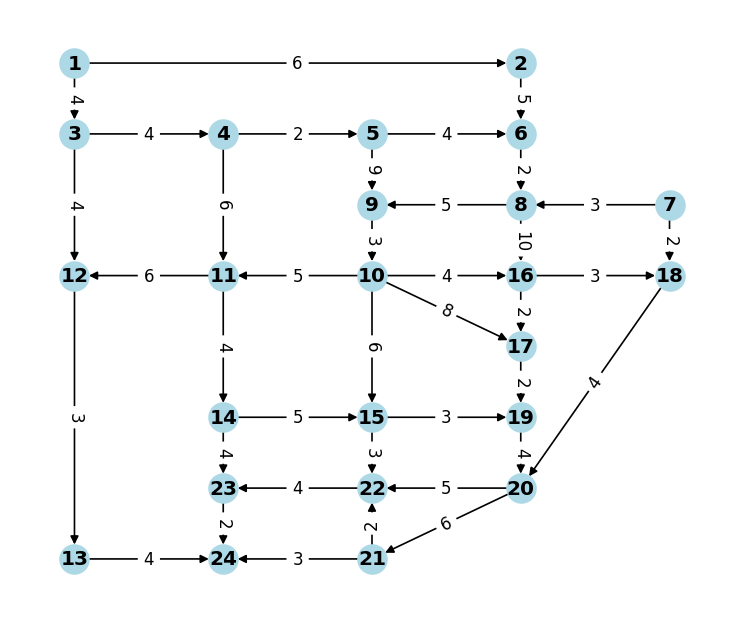}
    \caption{A more complex example network $2$ (based on the network structure of Sioux Falls) for our numerical study. The value on each edge denotes the free-flow road travel time.}
    \label{fig:network2}
\vspace{-3mm}
\end{figure}

\subsubsection{Uniform recommendation} We refer to the case of uniformly randomly picking a feasible path as a uniform recommendation. More specifically, when a user submits a recommendation request with the OD pair, the recommendation system randomly picks a path from the feasible path set to the user uniformly, regardless of current traffic conditions, other users, or other drivers.

\subsubsection{Without MNL Drivers} In this case, the recommendation system does not take other drivers modeled by the Multinomial Logit (MNL) model in \eqref{eq:MNL} into account. When suggesting the recommendations, this baseline underestimates the traffic conditions caused by other drivers who do not use the NRS. However, the travel time costs experienced by users are still influenced by those drivers.

\subsubsection{Proposed recommendation without updates (WU)} In this case, the NRS plans the proposed recommendation when the user submits the recommendation request with the OD pair. The recommendation persists unchanged throughout the user's journey, even in the event of an accident or significant changes in traffic conditions.

\subsubsection{Proposed updated recommendation (PU and RU)} In this case, the NRS plans the proposed recommendation when the user submits the recommendation request with the OD pair. The parallel recommendation updates according to the distributed gradient algorithm in \eqref{eq:update_P_u} is referred to as \textbf{PU}; the random updates according to \eqref{eq:update_P_u_r} is called \textbf{RU}. In this context, when an accident happens, or there are significant changes in traffic conditions, the users can adapt their paths promptly.

\begin{table}[htbp]
\caption{Case Study Setup for Network 1}
\begin{center}
\begin{tabular}{ccl}
\toprule
OD pair & User \# & Feasible paths \\
\midrule
$1$-$7$ & $1$ to $30$ & \thead[l]{Path $0$: $1$-$2$-$5$-$6$-$7$ \\ Path $1$: $1$-$2$-$7$ \\ Path $2$: $1$-$4$-$5$-$6$-$7$} \\
\midrule
$3$-$9$ & $31$ to $50$ & \thead[l]{Path $0$: $3$-$4$-$5$-$6$-$9$ \\ Path $1$: $3$-$4$-$5$-$8$-$9$ \\ Path $2$: $3$-$8$-$9$} \\
\midrule
$1$-$9$ & $51$ to $60$ & \thead[l]{Path $0$: $1$-$2$-$5$-$6$-$9$ \\ Path $1$: $1$-$2$-$5$-$8$-$9$ \\ Path $2$: $1$-$4$-$5$-$6$-$9$ \\ Path $3$: $1$-$4$-$5$-$8$-$9$} \\
\bottomrule \\[-0.3em]
\multicolumn{3}{c}{BPR cost with $\eta=0.35$, $\zeta=1$, $k_e=60, \forall e \in \mathcal{E}$.}\\
\multicolumn{3}{c}{Other drivers: $25$ with OD $1$-$9$ and $25$ with OD $3$-$7$.}
\end{tabular}
\end{center}
\label{tab:setup1}
\end{table}

\begin{table}[htbp]
\caption{Case Study Setup for Network 2}
\begin{center}
\begin{tabular}{ccl}
\toprule
OD pair & User \# & Feasible paths \\
\midrule
$1$-$10$ & $1$ to $60$ & \thead[l]{Path $0$: $1$-$2$-$6$-$8$-$9$-$10$ \\ Path $1$: $1$-$3$-$4$-$5$-$6$-$8$-$9$-$10$ \\ Path $2$: $1$-$3$-$4$-$5$-$9$-$10$} \\
\midrule
$2$-$18$ & $61$ to $85$ & \thead[l]{Path $0$: $2$-$6$-$8$-$9$-$10$-$16$-$18$ \\ Path $1$: $2$-$6$-$8$-$16$-$18$} \\
\midrule
$4$-$16$ & $86$ to $100$ & \thead[l]{Path $0$: $4$-$5$-$6$-$8$-$9$-$10$-$16$ \\ Path $1$: $4$-$5$-$6$-$8$-$16$ \\ Path $2$: $4$-$5$-$9$-$10$-$16$} \\
\bottomrule \\[-0.3em]
\multicolumn{3}{c}{BPR cost with $\eta=0.15$, $\zeta=4$, $k_e=100, \forall e \in \mathcal{E}$}\\
\multicolumn{3}{c}{Other drivers: OD $1$-$10$, $2$-$18$, $4$-$16$, $6$-$21$, $11$-$20$ are all $40$.}
\end{tabular}
\end{center}
\label{tab:setup2}
\end{table}

\subsection{Case study on network 1}

Consider network $1$ depicted in Fig. \ref{fig:network1}. In this scenario, we have $30$ users seeking to travel from node $1$ to $7$ (OD $1$-$7$), $20$ users from node $3$ to $9$ (OD $3$-$9$), and $10$ users from node $1$ to $9$ (OD $1$-$9$). The feasible path set for users is specified in Table \ref{tab:setup1}. Given the OD pairs and the network structure, the road-path incidence matrix satisfies the conditions $\sum_{i=1}^{k_u} a_{es_{u, i}}\leq 2, \ \forall e \in \mathcal{E}$.
Using the BPR function for the cost $c_e(\cdot)$ on each road $e \in \mathcal{E}$, and setting parameters $\zeta=1$, $\eta=0.35$, and $k_e=60$ for simplicity, the cost function $F_u$ for user $u$ evaluated at each time step is continuously differentiable in $\textbf{P}$ and convex in $\textbf{P}_u$. Additionally, besides the users, there are $25$ drivers with OD $1$-$9$ and $25$ drivers with OD $3$-$7$, whose behavior, following the MNL choice model, also contributes to the flow on roads in the network. 

\subsubsection{Scenario $1$: stable traffic conditions}
Under stable traffic conditions, free from lane closures, accidents, or abrupt alterations in traffic flow induced by other drivers, the local recommendation app for each user $u$ employs the parallel gradient update algorithm outlined in \eqref{eq:update_P_u} at each time step $n$ to adjust the recommended mixed strategy. The recommendations under scenario $1$ to one of the users with each OD pair are shown in Fig. \ref{fig:DRS_ex1}. To evaluate the effectiveness of these recommendations, we compare the travel time costs experienced by users for each OD pair to the baselines, and list them in Table \ref{tab:network1}.

\begin{figure}
    \centering
    \includegraphics[width=2.7in]{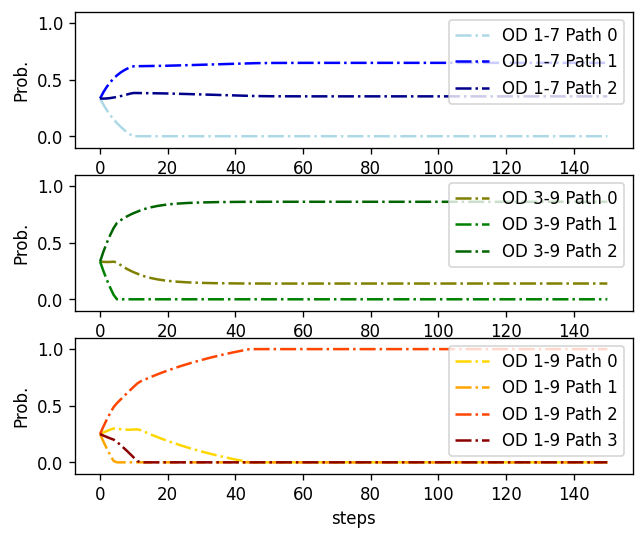}
    \caption{The proposed recommendations (PU) under stable traffic conditions for case study on network $1$.}
    \label{fig:DRS_ex1}
\vspace{-3mm}
\end{figure}

\begin{table}[htbp]
\vspace{+3mm}
\caption{Comparison of travel time for network $1$}
\begin{center}
\begin{tabular}{lcccccc}
\toprule
\multicolumn{7}{c}{Stable Traffic Conditions}\\
\midrule
 OD & Selfish & Uniform &  w/o MNL & WU & PU & RU\\
\midrule
      $1$-$7$ & 36.53 & 33.12 & 33.90 & \textbf{32.55}  & \textbf{32.55} & \textbf{32.55}\\ 
      $3$-$9$ & 39.60 & 36.37 & 36.88 & \textbf{35.42}  & \textbf{35.42} & \textbf{35.42}\\
      $1$-$9$ & 37.71 & 35.56 & 36.30 & \textbf{33.76}  & \textbf{33.76} & \textbf{33.76}\\
      Total & 2264.8 & 2076.7 & 2117.6 & \textbf{2022.5} & \textbf{2022.5} & \textbf{2022.5}\\
\midrule
\multicolumn{7}{c}{Changing Traffic Conditions}\\
\midrule
 OD & Selfish & Uniform &  w/o MNL & WU & PU & RU\\
\midrule
      $1$-$7$ & 40.59 & 35.55 & 34.41 & \textbf{33.79} & 34.07 & 34.07\\ 
      $3$-$9$ & 43.66 & 37.59 & 36.81 & \textbf{35.92} & 36.25 & 36.25\\
      $1$-$9$ & 41.77 & 37.39 & 38.25 & 37.28 & \textbf{35.38} & \textbf{35.38}\\
      Total & 2508.5 & 2192.2 & 2150.9 & 2104.8 & \textbf{2100.9} & \textbf{2100.9}\\
\bottomrule \\[-0.3em]
\end{tabular}
\end{center}
\label{tab:network1}
\vspace{-5mm}
\end{table}

In the case of selfish recommendation, users with OD pair $1$-$7$ are directed through path $2$, users with OD pair $3$-$9$ through path $0$, and users with OD pair $1$-$9$ solely through path $2$. While such a recommendation may be sufficient for a small number of users, for a larger number of users, recommending the same path to a significant fraction of users can undermine the effectiveness of the initially shortest route. This is called the flash crowd effect. The selfish shortest path recommendations result in a total travel time cost (summing up the expected cost for all the users) of $2264.8$ (units). However, by adhering to the recommendations depicted in Fig. \ref{fig:DRS_ex1}, the total travel time costs across all users can be reduced to $2022.5$ (units). The comparison tells us that the proposed recommendations suggested by the NRS are capable of mitigating the flash crowd effect. Besides, comparing the cost for each user as well as the total travel time cost between uniform recommendations and the proposed recommendations by the NRS, we can further conclude that strategic recommendation is essential in urban transportation networks.

In addition, by comparing the results between the recommendation without considering the other drivers (modeled by MNL) and the proposed recommendation suggested by the NRS, it becomes evident that accounting for non-user drivers is crucial. Failing to incorporate their impact may result in imprecise traffic conditions considered by the recommendation system, consequently leading to higher expected costs for users and an increased total cost of $2117.6$. This without MNL drivers scenario can be interpreted as a misinformed case for the NRS, wherein the recommendation system neglects the impact of drivers, resulting in non-precise traffic conditions.

Furthermore, it is worth mentioning that when implementing the proposed recommendations, users and drivers with OD pair $1$-$9$ incur costs of $33.76$ and $34.03$, respectively. This serves as further evidence of the importance of strategic recommendations. The better performance of the proposed recommendations compared to allowing drivers to choose paths based on their behaviors emphasizes the value of strategic navigational guidance.

\begin{figure}[!ht]
    \centering
    \begin{subfigure}[t]{0.48\textwidth}
        \centering
        \includegraphics[width=2.7in]{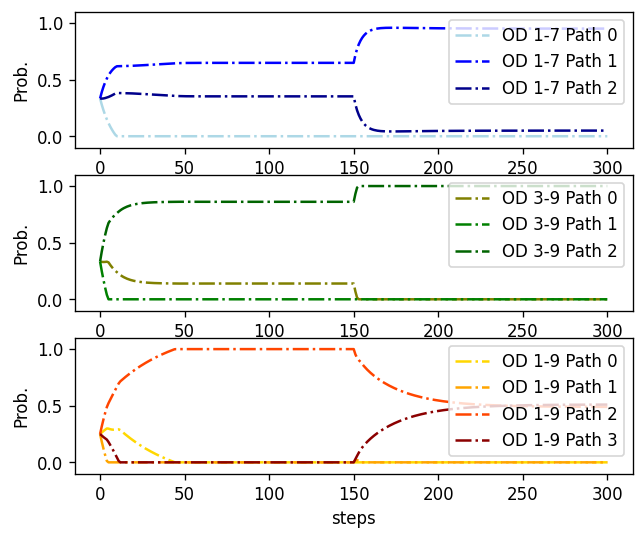}
        \caption{The parallel update recommendations (PU).}
        \label{fig:DRS_ex1_p}
    \end{subfigure}
    \hfill
    \begin{subfigure}[t]{0.48\textwidth}
        \centering
        \includegraphics[width=2.7in]{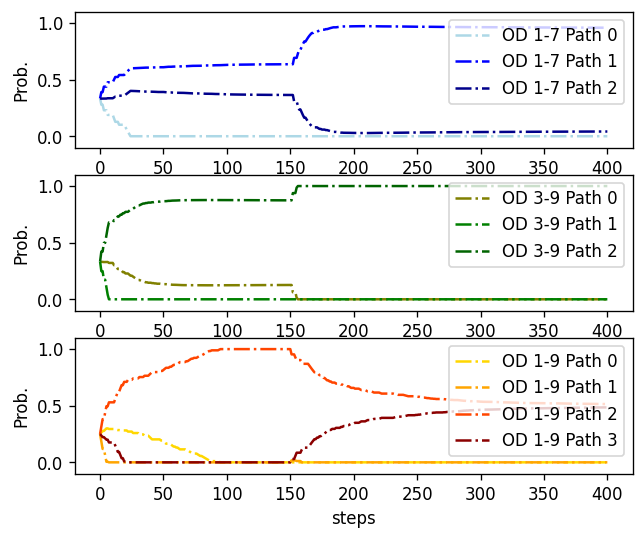}
        \caption{The random update recommendations (RU) with $\pi_u=0.5, \forall u$.}
        \label{fig:DRS_ex1_rand_}
    \end{subfigure}
    \caption{The updated recommendations under sudden changing traffic conditions for case study $1$ (An accident occurs on edge $5$-$6$ that caused an increase in travel time at time step $150$).}
\label{fig:DRS_ex1_}
\vspace{-3mm}
\end{figure}

\subsubsection{Scenario $2$: changing conditions}
As traffic conditions are dynamic and subject to rapid changes, such as accidents or sudden increases in flow due to other drivers, we need to consider these dynamic variations when framing the recommendation system. Our aim for this scenario is to assess if our updated scheme can effectively address these evolving traffic conditions. To illustrate, we consider an accident occurring on edge $5$-$6$, resulting in increased travel time on this road at time step $150$. The outcomes of both updated schemes are presented in Fig. \ref{fig:DRS_ex1_}, with cost comparisons against other baselines detailed in Table \ref{tab:network1}.

However, while the `with update' schemes (PU and RU) achieve a lower total travel time cost and benefit users with OD pair $1$-$9$, the costs for OD $1$-$7$ and $3$-$9$ increase compared to the `without update' scheme (WU). This shift occurs because users with OD $1$-$9$ are more severely affected by the accident, and the redistribution of flow caused by users with OD pair $1$-$9$ to alternative paths impacts the costs experienced by users with other OD pairs. This phenomenon is referred to as the ``paradox of route update''. It is important to note that in practice, the likelihood of encountering this paradox can be low. Users with OD pair $1$-$9$ would likely deviate from the recommendations generated by the `without update' scheme if they were aware of the accident. Consequently, the costs for OD pair $1$-$7$ and $3$-$9$ in the `without update' scheme would not be achievable. 

When comparing the parallel and random update schemes depicted in Fig. \ref{fig:DRS_ex1_p} and Fig. \ref{fig:DRS_ex1_rand_}, respectively, we note that their expected costs for users and total travel time costs are close, but the random update scheme takes more steps to stabilize, suggesting that the parallel update method is more effective. This shows that the random update strategy remains viable, especially in situations where implementing parallel updates is not feasible.


\subsection{Case study on network $2$} 
In this case, we consider a more practical scenario, where we adopt the network structure from the Sioux Falls network \cite{leblanc1975efficient}, and use parameters $\eta=0.15, \zeta=4$ for the BPR cost function as often recommended. It is worth noting that under such a setting the conditions $\sum_{i=1}^{k_u} a_{es_{u, i}}\leq 2, \ \forall e \in \mathcal{E}$ for the incidence-path matrix may not hold, which implies that the cost function $F_u$ for user $u$ evaluated at each time step is not necessarily convex in $\textbf{P}_u$. However, our experimental results still demonstrate that the gradient update algorithm in \eqref{eq:update_P_u} leads to stabilized recommendations if traffic conditions remain unchanged. The detailed setups for users and drivers are listed in Table \ref{tab:setup2}. For network $2$, we will examine both scenarios of stable and changing conditions, presenting the resulting travel time costs in Table \ref{tab:network2}.

\begin{table}[htbp]
\vspace{+3mm}
\caption{Comparison of travel time for network $2$}
\begin{center}
\begin{tabular}{lcccccc}
\toprule
\multicolumn{7}{c}{Stable Traffic Conditions}\\
\midrule
 OD & Selfish   & Uniform   &  w/o MNL   & WU   &  PU  &  RU\\
\midrule
      $1$-$10$ & 36.22 & 31.28 & 36.00 & \textbf{26.86}  & \textbf{26.86} & \textbf{26.86}\\ 
      $2$-$18$ & 32.92 & 30.60 & 32.19 & \textbf{24.82}  & \textbf{24.82} & \textbf{24.82}\\
      $4$-$16$ & 24.00 & 26.01 & \textbf{22.39} & 22.64  & 22.64 & 22.64\\
      Total & 3356.4 & 3033.0 & 3300.3 & \textbf{2571.7} & \textbf{2571.7} & \textbf{2571.7}\\
\midrule
\multicolumn{7}{c}{Changing Traffic Conditions}\\
\midrule
 OD & Selfish & Uniform &  w/o MNL & WU & PU & RU\\
\midrule
      $1$-$10$ & 45.55 & 36.13 & \textbf{26.40} & 27.45 & 26.57 & 26.57\\ 
      $2$-$18$ & 42.26 & 40.97 & 30.11 & \textbf{29.15} & 29.76 & 29.76\\
      $4$-$16$ & 31.70 & 35.07 & 28.38 & 28.25 & \textbf{28.05} & 28.06\\
      Total & 4265.2 & 3718.1 & 2762.3 & 2799.7 & \textbf{2759.0} & 2759.2\\
\bottomrule \\[-0.3em]
\end{tabular}
\end{center}
\label{tab:network2}
\vspace{-3mm}
\end{table}

\subsubsection{Scenario $1$: stable traffic conditions} 
The results under stable traffic conditions are shown in Fig. \ref{fig:DRS_ex2}, in which we can observe stabilized recommendations after time step $20$. Comparing the costs of selfish and uniform recommendations with the proposed one suggested by the NRS in Table \ref{tab:network2} drives us to the same conclusion as the case study for network $1$: strategic recommendation plays an important role in urban transportation networks.

While the lowest total travel time cost is attained by the proposed recommendations, the scenario without considering other drivers (modeled by MNL) surpasses the proposed recommendation for users with OD pair $4$-$16$, as indicated in Table \ref{tab:network2}. This ``paradox of misinformation'' arises because occasionally, inaccurate traffic conditions considered by the NRS may lead to lower costs for certain users. However, it can be argued that this situation may not occur in practice. If users have access to precise traffic conditions, recommendations resulting from misinformed traffic conditions (neglecting the impact of other drivers) would not align with the incentives of users other than OD pair $4$-$16$. In such cases, other users would not follow the recommendations, making the cost for OD pair $4$-$16$ in the table unattainable.

\begin{figure}
    \centering
    \includegraphics[width=2.7in]{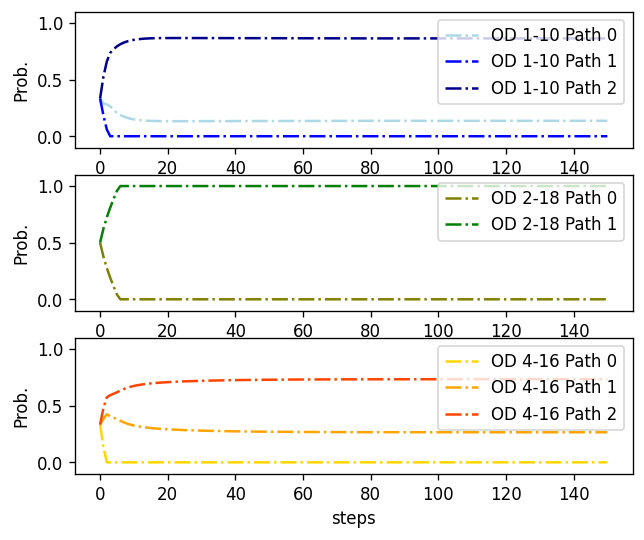}
    \caption{The recommendations (PU) under stable traffic conditions for case study on network $2$.}
    \label{fig:DRS_ex2}
\end{figure}

\begin{figure}[!ht]
    \centering
    \begin{subfigure}[t]{0.48\textwidth}
        \centering
        \includegraphics[width=2.7in]{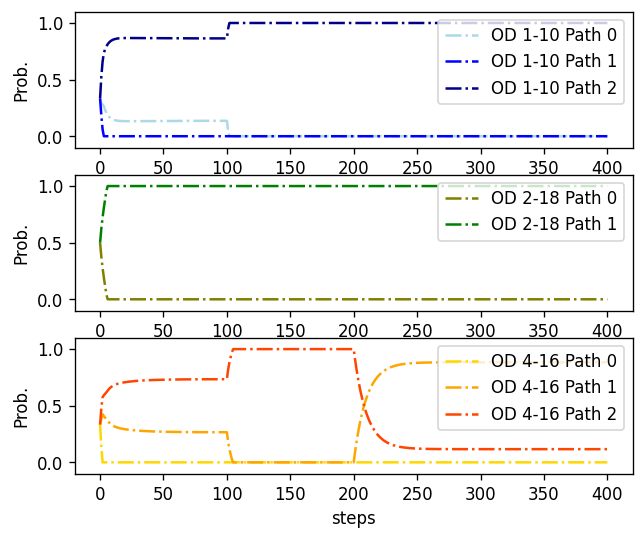}
        \caption{The parallel update recommendations (PU).}
        \label{fig:DRS_ex2_p}
    \end{subfigure}
    \hfill
    \begin{subfigure}[t]{0.48\textwidth}
        \centering
        \includegraphics[width=2.7in]{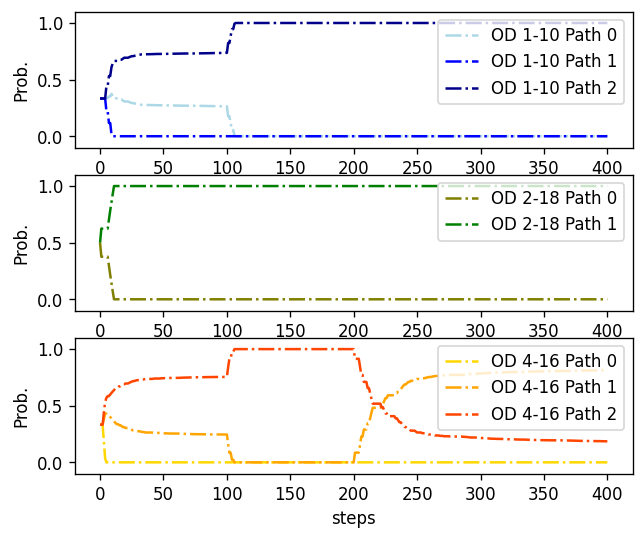}
        \caption{The random update recommendations (RU) with $\pi_u=0.5, \forall u$.}
        \label{fig:DRS_ex2_rand_}
    \end{subfigure}
    \caption{The updated recommendations under sudden changing traffic conditions for case study $2$ (An increase in travel time on edge $6$-$8$ at time step $100$ and an increase in travel time on edge $10$-$16$ at time step $200$).}
\label{fig:DRS_ex2_}
\vspace{-3mm}
\end{figure}

\subsubsection{Scenario $2$: changing conditions}
In this scenario, we assume an increase in travel time on edge $6$-$8$ at time step $100$, followed by a similar rise on edge $10$-$16$ at time step $200$. The results for parallel and random update schemes are shown in Fig. \ref{fig:DRS_ex2_p} and Fig. \ref{fig:DRS_ex2_rand_}, respectively. Despite encountering both the ``paradox of misinformation'' and the ``paradox of route update'' in this situation, the proposed parallel updated scheme consistently yields the lowest total travel time cost compared to other baselines. Furthermore, even in this more practical scenario, the random update strategy still works and stabilizes when traffic conditions remain unchanged, as evidenced in Fig. \ref{fig:DRS_ex2_rand_} after time step $300$.

\section{Conclusion}
In conclusion, our study introduces an incentive-compatible Navigation Recommendation System (NRS) from a game-theoretic perspective, addressing challenges posed by human user non-compliance and non-user driver behaviors in urban transportation networks. Our primary objective is to guide users toward an optimal traffic equilibrium, ensuring their compliance with the provided recommendations.
Moreover, we propose a dynamic NRS with parallel and random update schemes by leveraging V2X technology. This innovation enables users to receive real-time navigational updates while driving, fostering efficient navigation and empowering users to promptly adapt to changing traffic conditions. Our numerical analyses of two case studies demonstrate that the updated scheme achieves the lowest total user travel time costs while ensuring user engagement and adherence to recommendations.

\bibliography{reference}
\bibliographystyle{IEEEtran}

\appendix

\subsection{Proof of Proposition \ref{prop:converge}} \label{app:prop_converge}
We divide the proof into the following five steps.
\vspace{2mm}

\noindent \textbf{Step 1:} \textit{Let $\textbf{P} \in \mathcal{P}$ and $\textbf{q} \in \mathbb{R}^d$, where $d=\Pi_{u \in \mathcal{U}}|\mathcal{S}_u|$. Then, $(\text{proj}_{\mathcal{P}}[\textbf{q}]-\textbf{q})^\top (\text{proj}_{\mathcal{P}}[\textbf{q}]-\textbf{P}) \leq 0$, $\forall \textbf{P} \in \mathcal{P}$.}

\begin{proof}
    By the definition of $\text{proj}_{\mathcal{P}}[\textbf{q}]:=\argmin_{\textbf{P} \in \mathcal{P}} \frac{1}{2}|| \textbf{P}- \textbf{q} ||^2$ for $\textbf{q} \in \mathbb{R}^d$ and the gradient of $\frac{1}{2}|| \textbf{P}- \textbf{q} ||^2$ at $\textbf{P}=\text{proj}_{\mathcal{P}_u}[\textbf{q}]$ being $(\text{proj}_{\mathcal{P}}[\textbf{q}]-\textbf{q})$, the statement is the first-order optimality condition for $\text{proj}_{\mathcal{P}}[\textbf{q}]$.
\end{proof}

\noindent \textbf{Step 2:} \textit{For any time step $n$, $||\textbf{P}^{(n+1)}-\textbf{P}^*|| \leq ||\textbf{q}^{(n+1)}-\textbf{P}^*||$, where $\textbf{P}^* \in \mathcal{P}$ is an NE defined in Definition \ref{def:NE}.}

\begin{proof}
    With the fact that $\textbf{P}^{(n+1)}=\text{proj}_{\mathcal{P}}[\textbf{q}^{(n+1)}]$, the inequality specified in step 1 becomes $$\big(\textbf{P}^{(n+1)}-\textbf{q}^{(n+1)}\big)^\top \big(\textbf{P}^{(n+1)}-\textbf{P}^*\big) \leq 0.$$ By expanding $\textbf{P}^{(n+1)}-\textbf{q}^{(n+1)}=\textbf{P}^{(n+1)}-\textbf{P}^*+\textbf{P}^*-\textbf{q}^{(n+1)}$, the inequality indicates that
    \begin{align*}
        ||\textbf{P}^{(n+1)}-\textbf{P}^*||^2 &\leq \big(\textbf{q}^{(n+1)}-\textbf{P}^*\big)^\top \big(\textbf{P}^{(n+1)}-\textbf{P}^*\big)\\
        &\leq ||\textbf{q}^{(n+1)}-\textbf{P}^*|| \ ||\textbf{P}^{(n+1)}-\textbf{P}^*||,
    \end{align*} where the last inequality comes from the Cauchy-Schwarz inequality. Then, dividing both sides by $||\textbf{P}^{(n+1)}-\textbf{P}^*||$ completes the proof.
\end{proof} 

\noindent \textbf{Step 3:} \textit{By letting $D_{n+1}=\frac{1}{2}||\textbf{P}^{(n+1)}-\textbf{P}^*||^2$, then $D_{n+1} \leq D_{n} + \frac{\alpha_n^2}{2} ||\nabla F(\textbf{P}^{(n)}; \textbf{P}_o^{(n)})||^2.$}

\begin{proof}
    First, recall that at the same point $\textbf{P}_{u}^{(n)}, \textbf{P}_{-u}^{(n)}, \textbf{P}_o^{(n)}$, $\nabla_u \Tilde{F}_u(\textbf{P}_u^{(n)},\textbf{m}^{(n)})$ and $ \nabla_u F_u(\textbf{P}_{u}^{(n)}, \textbf{P}_{-u}^{(n)}; \textbf{P}_o^{(n)})$ have an identical value, then $\nabla \Tilde{F}(\textbf{P}^{(n)}, \textbf{m}^{(n)})$ and $\nabla F(\textbf{P}^{(n)}; \textbf{P}_o^{(n)})$ also have an identical value at the same $\textbf{P}^{(n)}, \textbf{P}_o^{(n)}$. Then, following step 2, we have
    \begin{align*}
    &||\textbf{P}^{(n+1)}-\textbf{P}^*||^2 \leq ||\textbf{q}^{(n+1)}-\textbf{P}^*||^2\\
    & = ||\textbf{P}^{(n)}-\alpha_n \nabla \Tilde{F}(\textbf{P}^{(n)}, \textbf{m}^{(n)}) -\textbf{P}^*||^2\\
    & = ||\textbf{P}^{(n)}-\alpha_n \nabla F(\textbf{P}^{(n)}; \textbf{P}_o^{(n)}) -\textbf{P}^*||^2\\
    &=||\textbf{P}^{(n)}-\textbf{P}^*||^2-2 \alpha_n \nabla F(\textbf{P}^{(n)}; \textbf{P}_o^{(n)})^\top \big(\textbf{P}^{(n)}-\textbf{P}^*\big) \\
    & \qquad \qquad \qquad \quad + \alpha_n^2 ||\nabla F(\textbf{P}^{(n)}; \textbf{P}_o^{(n)})||^2\\
    &\leq ||\textbf{P}^{(n)}-\textbf{P}^*||^2 + \alpha_n^2 ||\nabla F(\textbf{P}^{(n)}; \textbf{P}_o^{(n)})||^2,
\end{align*} where the third equality follows by the conditions in Proposition \ref{prop:converge} so that $\nabla F(\textbf{P}^{(n)}; \textbf{P}_o^{(n)})^\top \big(\textbf{P}^{(n)}-\textbf{P}^*\big) \geq 0$. Dividing both sides by $\frac{1}{2}$, we can obtain $D_{n+1} \leq D_{n} + \frac{\alpha_n^2}{2} ||\nabla F(\textbf{P}^{(n)}; \textbf{P}_o^{(n)})||^2$, which completes the proof.
\end{proof}

\noindent \textbf{Step 4:} \textit{Define $S_n = D_n + \frac{L}{2} \sum_{t=n}^{\infty} \alpha_t^2$, then $S_{n+1} \leq S_{n}$.}

\begin{proof}
    According to step 3 and conditions in Proposition \ref{prop:converge} so that for any time step $n$, $||\nabla F(\textbf{P}^{(n)}; \textbf{P}_o^{(n)})||^2 \leq L$, we have \begin{align*}
    S_{n+1} &= D_{n+1} + \frac{L}{2} \sum_{t=n+1}^{\infty} \alpha_t^2\\
    & \leq D_n + \frac{L}{2} \alpha_n^2 + \sum_{t=n+1}^{\infty} \alpha_t^2 = D_n + \frac{L}{2} \sum_{t=n}^{\infty} \alpha_t^2 = S_n,
\end{align*} which shows that $S_{n+1} \leq S_{n}$.
\end{proof}

\noindent \textbf{Step 5:} \textit{Under the assumption that $\sum_{n=1}^{\infty} \alpha_n^2 < \infty$, the sequence of ${\textbf{P}^{(n)}}$ converges to $\textbf{P}^*$.}

\begin{proof}
    By applying the argument in step 4 recursively, with the common assumption that $\sum_{n=1}^{\infty} \alpha_n^2 < \infty$, we have
$$
    S_{n} \leq D_1 + \frac{L}{2} \sum_{n=1}^{\infty} \alpha_n^2 < \infty, 
$$ indicating that the sequence $S_n$ converges. The definition of $S_n = D_n + \frac{L}{2} \sum_{t=n}^{\infty} \alpha_t^2$ thereby implies the convergence of the sequence $D_n$. Consequently, the sequence ${\textbf{P}^{(n)}}$ converges to $\textbf{P}^*$, which completes the proof.
\end{proof}

\subsection{Proof of Proposition \ref{prop:converge2}} \label{app:prop_converge2}
The proof extends the argument in Appendix \ref{app:prop_converge} used for the parallel update algorithm. For the $u$-th component of $\textbf{P}^{(n+1)}$ in \eqref{eq:update_P_u_r}, with the conditions stated in Proposition \ref{prop:converge2}, we have an analogy for step 3 of Appendix \ref{app:prop_converge}.
\vspace{2mm}

\noindent \textbf{Step 3$^\prime$:} \textit{Let $D_u^{(n+1)}=\frac{1}{2}||\textbf{P}_u^{(n+1)}-\textbf{P}_u^*||^2$, then $\E[D_u^{(n+1)}] \leq E[D_u^{(n)}] + \pi_u \frac{\alpha_n^2}{2} ||\nabla_u F_u(\textbf{P}_u^{(n)}, \textbf{P}_{-u}^{(n)}, \textbf{P}_o^{(n)})||^2.$}

\begin{proof}
    Recall that the recommendation for user $u$ updates with probability $0 < \pi_u <1$, 
    \begin{align*}
    \E[D_u^{(n+1)}] & =\pi_u \big(\E[D_u^{(n)}] +  \frac{\alpha_n^2}{2} ||\nabla_u F_u(\textbf{P}_u^{(n)}, \textbf{P}_{-u}^{(n)}, \textbf{P}_o^{(n)})||^2\big) \\
    & \quad + (1- \pi_u)\E[D_u^{(n)}]\\
    &=\mathbb{E}[D_u^{(n)}] + \pi_u \frac{\alpha_n^2}{2} ||\nabla_u F_u(\textbf{P}_u^{(n)}, \textbf{P}_{-u}^{(n)}, \textbf{P}_o^{(n)})||^2,
\end{align*} which completes the proof.
\end{proof}

\noindent \textbf{Step 4$^\prime$:} \textit{Define $S_u^{(n)} = D_u^{(n)} + \frac{L^\prime}{2} \sum_{t=n}^{\infty} \alpha_t^2$, then $\E[S_u^{(n+1)}]\leq \E[S_u^{(n)}], \forall u \in \mathcal{U}.$}

\begin{proof}
    According to step 3$^\prime$ and conditions in Proposition \ref{prop:converge2} so that for any time step $n$, $||\nabla_u F_u(\textbf{P}_u^{(n)}, \textbf{P}_{-u}^{(n)}; \textbf{P}_o^{(n)})||^2 \leq L^\prime$, we have \begin{align*}
    \E[S_u^{(n+1)}] &= \E[D_u^{(n+1)}] + \frac{L^\prime}{2} \sum_{t=n+1}^{\infty} \alpha_t^2\\
    &= \E[D_u^{(n)}] + \pi_u \frac{L^\prime}{2} \alpha_n^2 + \frac{L^\prime}{2} \sum_{t=n+1}^{\infty} \alpha_t^2\\
    &\leq \E[D_u^{(n)}] + \frac{L^\prime}{2} \sum_{t=n}^{\infty} \alpha_t^2 = \E[S_u^{(n)}],
\end{align*} which shows that $\E[S_u^{(n+1)}]\leq \E[S_u^{(n)}]$.
\end{proof}

\noindent \textbf{Step 5$^\prime$:} \textit{The sequence $\textbf{P}_u^{(n)}$ converges a.s. to $\textbf{P}_u^*, \forall u \in \mathcal{U}$.}

\begin{proof}
    Using the notion of $l_\infty$-norm and denoting $||S^{(n)}||_\infty= \max_u \E[S_u^{(n)}]$, step 4$^\prime$ leads us to $||S^{(n+1)}||_\infty \leq ||S^{(n)}||_\infty$, which can be written as $||S^{(n+1)}||_\infty \leq \gamma ||S^{(n)}||_\infty$ with some  $\gamma \in [0, 1]$. Applying the argument recursively, we can get $||S^{(n)}||_\infty \leq \gamma^n ||S^{(0)}||_\infty$. By the Markov inequality, 
\begin{align*}
    \sum_{n=1}^{\infty} Pr(S_u^{(n)} & \geq \epsilon) \leq \sum_{n=1}^{\infty} \frac{\E[S_u^{(n)}]}{\epsilon} \leq \frac{1}{\epsilon}\sum_{n=1}^{\infty}||S^{(n)}||_\infty\\
    & \leq \frac{1}{\epsilon}\sum_{n=1}^{\infty} \gamma^n||S^{(0)}||_\infty = \frac{1}{\epsilon (1-\gamma)} ||S^{(0)}||_\infty,
\end{align*} which means that $\sum_{n=1}^{\infty} Pr(S_u^{(n)} \geq \epsilon) < \infty$ when $\epsilon > 0, \forall u$. Then, for all $u \in \mathcal{U}$, $$Pr\left(\limsup_{n \to \infty}\{\omega: S_u^{(n)} \geq \epsilon\}\right)=0$$ by the Borel-Cantelli lemma. This shows that sequence $S_u^{(n)}$ converges almost surely, thereby indicating the almost surely convergence of sequence $D_u^{(n)}$. As a result, sequence $\textbf{P}_u^{(n)}$ converges almost surely to $\textbf{P}_u^*, \forall u \in \mathcal{U}$, which then completes the proof.
\end{proof}


\end{document}